\documentclass[authoryear,12pt]{article}

\RequirePackage[OT1]{fontenc}
\RequirePackage{amsthm,amsmath}
\RequirePackage{natbib} 
\RequirePackage[colorlinks,citecolor=blue,urlcolor=blue]{hyperref}

\usepackage{bm}
\usepackage[english]{babel}
\usepackage{latexsym}
\usepackage{amssymb}
\usepackage{graphicx}
\usepackage[normalem]{ulem}
\usepackage{xr}

\numberwithin{equation}{section}
\theoremstyle{plain}
\newtheorem{theorem}{Theorem}[section]
\newtheorem{proposition}{Proposition}[section]
\newtheorem{corollary}{Corollary}[section]
\newtheorem{lemma}{Lemma}[section]
\newtheorem{assumption}{Assumption}

\newtheorem{example}{Example}[section]

\newcommand{\nesq}{\subseteq_{\emptyset}}
\newcommand{\nes}{\subset_{\emptyset}}
\newcommand{\I}{\mathcal{I}}
\newcommand{\Yobs}{Y^{\mathrm{obs}}}
\newcommand{\Ymis}{Y^{\mathrm{mis}}}
\newcommand{\YV}{Y^{{\cal{V}}}}

\begin{document}
	\title{Causal inference for binary non-independent outcomes}
	\author{Monia Lupparelli$^\dag$ and Alessandra Mattei$^\ddag$\\
		$^\dag$University of Bologna and  $^\ddag$University of Florence\\ 	
		$^\dag$email: monia.lupparelli@unibo.it\\
		  $^\ddag$email: mattei@disia.unifi.it}
	\date{\today}
	\maketitle

\begin{abstract}
Causal inference on multiple non-independent outcomes  raises  serious challenges,
because multivariate techniques that properly account for the outcome's dependence structure need to be considered.
We focus on the case of binary outcomes framing our discussion in the potential outcome approach to causal inference.
We define causal effects of treatment on joint outcomes introducing  the notion of product outcomes.
We also discuss a decomposition of the causal effect on product outcomes into  marginal and  joint causal effects, which respectively provide information  on treatment effect  on the  marginal (product) structure of the product outcomes and on the outcomes' dependence structure.
We propose a log-mean linear regression approach for modeling the distribution of the potential outcomes, which is particularly appealing because all the causal estimands of interest and the decomposition into marginal  and joint causal effects can be easily derived  by model parameters.  The method is illustrated 
in two randomized experiments concerning (i)  the effect of the administration of  oral pre-surgery morphine on pain intensity  after  surgery; and  (ii)  the effect of honey on nocturnal cough and sleep difficulty associated with childhood upper respiratory tract infections.
\end{abstract}

\section{Introduction} 
Causal studies involving multivariate outcome variables are increasingly widespread in real-world applications: 
intervention studies in many fields routinely collect information on multiple outcomes. Recently, a strand of the causal inference literature has been working on using multiple outcomes, possible coupled with conditional independence assumptions, to address identification problems in causal studies with intermediate variables \citep{MLM:2013, MealliPacini:2013, MercatantiLiMealli:2015, MealliPaciniStanghellini:2015}. 
In these studies focus is on causal effects on a single response variable, which is viewed as the outcome of main interest, and additional outcomes are used as auxiliary variables for inferential purposes. 

In this paper we consider a different type of studies, where   
focus is on causal effects on a multivariate response variable, and thus the whole  vector of outcomes is the response variable of main interest. Assessing causal effects on multivariate outcomes presents unique challenges, because  causal effects on joint sets of endpoint outcomes need to be properly defined,  and multivariate inferential methods for identifying  and estimating those causal estimands that also properly account for the outcomes' dependence structure need to be developed.

The existing  causal inference literature has rarely focused on assessing causal effects on multivariate response variables. Some exceptions include \cite{JoMuthen:2001}, who conduct 
a joint analysis with two outcomes  in the context of a randomized trial with noncompliance; \cite{HBR:2002}, who focus on estimating causal effects of a time-varying treatment on the mean of a repeated measures outcome using a marginal structural model; 
\cite{FlandersKlein:2015}, who propose a general definition of causal effects, showing how  it can be applied in the presence of multivariate outcomes to define  causal effects for   specific sub-populations of units or vector of causal effects; and \cite{LiPeng:2017}, who establish finite population central  limit theorems in completely randomized experiments where the response variable may be multivariate and  
causal estimands of interest are defined as linear combinations of the potential outcomes.

We focus on assessing causal effects of a treatment on multiple binary outcomes. 
The binary nature of the outcomes raises further challenges. When outcomes are binary the definition of a measure of association is tricky and requires to account for several critical aspects.   The dependence structure characterizing categorical variables is usually hard to investigate because pairwise associations do not provide a complete picture of it, but higher order associations need to be considered.   Also, exploring the parameter space, which consists of joint probabilities, is awkward especially because its dimension increases exponentially as the number of variables increases.  

The main  contribution of the paper consists in providing a novel and appealing framework for   drawing causal inference for binary multivariate outcomes. Specifically, we discuss and address the following issues. 
First, we formally define causal effects on multiple binary outcomes adopting the potential outcome approach to causal inference,  commonly referred to as Rubin's Causal Model \cite[RCM, e.g.,][]{Rubin:1974, Rubin:1977, Rubin:1978}. See also \cite{ImbensRubin:2015} for a comprehensive overview of the potential outcome approach. Specifically, we focus on causal relative risks, that is, ratio of probabilities of success corresponding to potential outcomes under different treatment conditions on a common set of units. 
Second, to formally define the causal estimands of interest,
we introduce new binary outcomes,  defined as function (product) of subsets of outcomes,  that we call product outcomes.
We  propose a decomposition of the causal relative risks for product outcomes  into two components: 
one representing causal effects on marginal outcomes, and the other representing causal effects on the outcomes' dependence structure. This decomposition may provide valuable information on how the treatment acts, revealing whether treatment effects on the multivariate outcome are mainly either through  treatment effects on (subsets of) marginal outcomes or through  treatment effects on the outcomes' dependence structure. 
Third, we propose to model the joint distribution of potential outcomes (conditional on a set of pre-treatment variables) using the class of log-mean linear regression models introduced by \citet{lup-rov:2016}. We generalize and extend results in \citet{lup-rov:2016} in order to properly account for the fact that here  models are specified for the potential outcomes, rather than observed outcomes.  We  show that  the model parameters are directly interpretable in term of the causal relative risks we are interested in and that they can be combined to derive a natural and easily interpretable decomposition of the causal relative risks for product outcomes.

We illustrate our framework in two medical examples, to which we refer as the morphine study and 
the honey study throughout the paper. The morphine study is a prospective, randomized, double-blind clinical study aimed at evaluating the effect of preoperative administration of oral morphine sulphate on postoperative pain relief \cite[see][for details]{Borracci2013}.   We use  this study to illustrate the key concepts throughout the paper.

The honey study  is a clinical randomized experiment aimed at evaluating the effect of buckwheat honey or honey-flavored dextromethorphan (an over-the-counter drug) versus no treatment on nocturnal cough and sleep difficulties associated with  childhood upper respiratory tract infections \citep{paul-al:2007}.

\section{Theoretical Framework} \label{sec:rcm}
\subsection{Basic setup}
Given a finite set $V=\{1,\dots,p\}$,  let $Y_V=(Y_v)_{v \in V}$ be the vector of binary outcomes of interest. Every single outcome takes level 1 in case of success, and level 0 in case of failure; then, the full vector $Y_V$ takes value $y_V \in \mathcal{I}_V=\{0,1\}^{|V|}$, where $|V|=p$ is the cardinality of the set $V$. For every $D \subseteq V$, $Y_D$ is a marginal vector of outcomes such that $Y_D=1_D$, if $Y_v=1$ for all $v \in  D$, $Y_D=0_D$, if $Y_v=0$ for all $v \in  D$, and it takes  any other value  $y_D \in \mathcal{I}_{D}=\{0,1\}^{|D|}$, with $y_D \neq 1_D,0_D$, otherwise; where  $1_D$ and $0_D$ are two vectors of 1s and 0s of size $|D|$.
For every multiple outcome $Y_D=(Y_v)_{v \in D}$ with $D \subseteq V$, we refer to the event $Y_D=1_D$ as a joint success, and to the event defined by any other level $y_D \neq 1_D$ as a  joint failure. Notice that, among the joint failure events, we do not distinguish between $Y_D=0_D$ and $Y_D=y_D$  for any $y_D \neq 1_D,0_D$, because both cases do not represent the event of interest, that is, a  joint success.

In this work we are interested in assessing effects of a treatment both on single variables, $Y_v$, $v \in V$, as well as on joint  variables $Y_D$ with $D \subseteq V$ and $|D|>1$. Specifically, we are interested in assessing treatment effects on the occurrence of a joint success $Y_D=1_D$, for every $D \subseteq V$. To this aim we make use of a new set of variables that we call $D$-product outcomes.

	Given a random vector $Y_V=(Y_v)_{v \in V}$ of binary outcomes, for every non-empty subset $D$ of $V$, 	the $D$-product outcome is defined as follows:
	\begin{equation}
	Y^D=\prod_{v \in D} Y_v.
	\end{equation}	

Let ${\cal{V}}$ be the power set of $V$ minus the empty-set and let $\YV=(Y^D)_{D \subseteq V, D \neq \emptyset}$  denote the vector of all product outcomes. For sake of simplicity, in the sequel we adopt the shorthand notation $D \nesq V$ and $D \nes V$ to denote any subset (or proper subset) $D$ of $V$ not equal to the empty-set.
For any pair $Y^D$ and $Y^{D'}$  with $D' \nes D$, we say that $Y^{D'}$ is a nested product outcome of $Y^D$. Note that each $Y^D$  is a binary variable which takes level 1 if $Y_D =1_D$, and level 0 otherwise, that is, $Y^D=1$ if and only if a joint success realizes for the outcome variable $Y_D$. Assessing treatment effects on $D$-product outcomes, $Y^D \in \YV$, represents the main focus of our work. See also \citet{lup-rov:2016} who adopt similar product variables for different purposes.

It is straightforward to figure out that, in case of non-independent outcomes, effects of the treatment on $Y^D$ cannot be investigated by only exploiting  information about treatment effects on single outcomes: the treatment may affect $Y^D$ both through its effect on each single outcome, $Y_v$, $v \in D$, as well as through its effect on the association structure among the variables belonging to $Y_D$. In order to formalize these concepts we need to define causal effects  introducing  a formal framework for causal inference.  We adopt the potential outcome approach to causal inference \citep[][]{Rubin:1974, Rubin:1977, Rubin:1978}. 

\subsection{Potential Outcomes}
Consider a group of units each of which can potentially be assigned to a binary treatment $w$, with $w = 1$ for active treatment and $w = 0$ for control. We take a super-population perspective,  considering the $n$ observed units  as a random sample from an infinite super-population.
Under the stable unit treatment value assumption \cite[SUTVA,][]{Rubin:1980}, which rules out both hidden versions of treatments  as well as interference between units, we can define for each outcome variable, $Y_v$, $v \in V$, two potential outcomes for each unit. 
Let $Y_{v}(0)$ denote the value of $Y_v$ under treatment $w=0$, and let $Y_{v}(1)$ denote
the value of $Y_v$  under treatment $w=1$. Let $Y_V(w)=(Y_v(w))_{v \in V}$  be the random vector including potential outcomes for every variable under treatment level $w$, $w=0,1$. 
Potential outcomes for $D$ product outcomes need to be introduced, too. 
Let $Y^{\cal{V}}(w)=(Y^D(w))_{D \nesq  V}$ be the random vector of $D$-product potential outcomes where, for every non-empty subset $D$ of $V$, 
\begin{equation}\label{eq:D-potential-outcome}
Y^D(w)=\prod_{v \in D} Y_v(w).
\end{equation}	

Every $Y^D(w)$ is a binary random variable which takes level 1 if $Y_D(w)=1_D$, and level 0 otherwise. For the special case with $|D|=1$,  the $D$-product potential outcome $Y^D(w)$ coincides with a potential outcome $Y_v(w)$, for a certain $v \in V$.  Then $\YV(w) =(Y^D(w))_{D \nesq V}$ is  the augmented vector combining the vector $Y_V(w)$ with  all the $D$-product potential outcomes $Y^D(w)$ for any  $D \nesq V$. 

\begin{example} Morphine study. Let  $Y_V(w)=(Y_1(w), Y_2(w))$ be a bivariate vector, with $Y_1(w)$ and $Y_2(w)$ denoting  pain intensity after surgery at rest and  on movement, respectively ($0=$ high; $1=$ low) under treatment level $w$, with $w=0$ for the placebo treatment and $w=1$ for the preoperative morphine treatment. Then, we have $\YV(w) = (Y_1(w), Y_2(w), Y^{\{1,2\}}(w))$, where $Y^{\{1,2\}}(w)=  Y_1(w)\cdot Y_2(w)$ is a binary variable equal to one for patients with a low level of pain intensity both at rest and  on movement under  treatment level $w$.
\end{example}

In our analysis we assume that a set of individual covariates is also available, which are collected in a vector $X_U$ with $U=\{1,\dots,q\}$ defining the finite set of indexes for the  covariates. In this context, without loss of generality, we consider binary covariates such that $X_U=x_U$, with $x_U \in \mathcal{I}_U=\{0,1\}^{q}$. Nevertheless the generalization for the inclusion of continuous covariates is conceptually straightforward.

\section{Causal estimads}\label{sec:causal}
\subsection{Causal relative risks}
In the potential outcome approach,  causal effects are defined as comparisons of potential outcomes under  different treatment levels for a common set of units. For instance, a causal effect of the treatment $w=1$ versus treatment $w=0$
on a single outcome $Y_v$ is defined as a comparison of the potential outcomes $Y_v(1)$ and $Y_v(0)$ on a
common set of units. 

In this paper we focus on  causal relative risks. The  causal relative risk for a specific outcome $Y_v$ is defined as follows:
\begin{equation}\label{eq:RRv}
RR_{v}=\dfrac{P[Y_v(1)=1]}{P[Y_v(0)=1]}, \qquad v \in V.
\end{equation}

Sometimes  the interest is on casual effects for specific sub-populations defined in terms of a set $X_U$ of covariates, that is, on conditional causal effects \cite[e.g.,][]{Imbens:2004, AtheyImbens:2015}.  For example, we may be interested in the causal relative risk of the morphine treatment 
on post-operative pain intensity on movement for male and female, separately. 
Then, the  relative risk of the treatment on an  outcome, $Y_v$, given a fixed level $x_U$ of the covariates, is  
\begin{equation}\label{eq:single-causal-effect}
RR_{v\mid x_U}=\frac{P[Y_v(1)=1\mid X_U=x_U]}{P[Y_v(0)=1\mid X_U=x_U]}, \qquad v \in V, \;\; x_U \in \mathcal{I}_U.
\end{equation}

Throughout the paper, we will focus on the causal relative risks in Equation \eqref{eq:single-causal-effect} and we define new causal estimands  conditional on  values $x_U \in \mathcal{I}_U$ of the covariate set $X_U$, because they may provide precious information on the effectiveness of the treatment across sub-populations defined by the values of the covariates.  Nevertheless marginal effects can be derived  marginalizing over $X_U$.

For any product outcome $Y^D$, let
\begin{equation}\label{eq:D-product-causal-effect}
RR_{D\mid x_U}=\frac{P[Y^D(1)=1\mid X_U=x_U]}{P[Y^D(0)=1\mid X_U=x_U]}, \qquad D  \subseteq V, \;\; x_U \in \mathcal{I}_U.
\end{equation}
be the $D$-product  relative risk for a given value, $x_U$, of a set $X_U$ of covariates.

In the special case when $|D|=1$, Equation \eqref{eq:D-product-causal-effect}   coincides with  Equation \eqref{eq:single-causal-effect}. Also, we adopt the convention $RR_{\emptyset\mid x_U}=1$,  so that, when $D$ is used to index the relative risk $RR_{D\mid x_U}$ rather than  a product outcome $Y^D$, we can avoid to specify $D \neq \emptyset$. 

\subsection{Marginal and joint causal effects}
For any product outcome $Y^D \in \YV$,  we  propose to distinguish between two different causal effects that we  call the marginal  effect and the joint  effect of the treatment. The former  accounts for the  effect deriving from the  product-structure of $Y^D$, which necessarily embodies information provided by causal effects on marginal product  outcomes $Y^{D'}$, for all $D' \nes D$.  The latter accounts for the  effect of the treatment on   the association structure of the joint distribution of $Y_D$.

It is reasonable to expect that  causal effects in Equation \eqref{eq:D-product-causal-effect} are  a combination of  marginal and joint effects. 
For instance,  the effect of the morphine treatment on post-operative pain intensity at rest and on movement, $Y^{\{1,2\}}$, combines the marginal effect of the  treatment on each single outcome with the joint effect on their association. 

In order to formally  address these concepts,  we introduce two additional causal estimands. Given any causal estimand  $\theta_{D\mid x_U}$ for the product-outcome $Y^D$, the  marginal causal effect ($MCE$)  is defined as
\begin{equation}
	MCE_{D\mid x_U}= h[(\theta_{D'\mid x_U})_{D' \subset D}],  \qquad D \subseteq V, \label{eq:ICE1} 
	\end{equation}
for a suitable function $h: \mathcal{R}^{2^{|D|}-1} \rightarrow \mathcal{R}$.
The joint causal effect  ($JCE$)  is defined as comparison  of an association measure $g(\cdot)$  between the joint distributions of $Y_D(1)$ and $Y_{D}(0)$:
	\begin{eqnarray}  \label{eq:ECE1}
\lefteqn{\qquad JCE_{D\mid x_U} : }\\& & \quad g[P(Y_{D}(1)\mid  X_U=x_U)] \quad versus \quad g[P(Y_{D}(0)\mid  X_U=x_U)], \quad D \subseteq V. \nonumber
	\end{eqnarray}


Investigating these two causal effects  represents an interesting issue, because they may provide useful insights on how the treatment acts.   To fix the ideas, suppose that for any vector $Y_D(w)$,  $D \subseteq V$, $w=0,1$,  the components $Y_v(w)$, with $v \in D$ are  mutual independent given the covariates. Then,  the  causal relative risk  for every $D \subseteq V$ is
\begin{eqnarray*}
	RR_{D\mid x_U}&=&\dfrac{P[Y^D(1)=1\mid X_U=x_U]}{P[Y^D(0)=1\mid X_U=x_U]} = \dfrac{P[\cap_{v \in D}Y_v(1)=1\mid X_U=x_U]}{P[\cap_{v \in D}Y_v(0)=1\mid X_U=x_U]}  \\&=& \dfrac{\prod_{v \in D} P[Y_v(1)=1\mid X_U=x_U]}{\prod_{v \in D} P[Y_v(0)=1\mid X_U=x_U]}=
	\prod_{v \in D} 	\dfrac{P[Y_v(1)=1\mid X_U=x_U]}{ P[Y_v(0)=1\mid X_U=x_U]}  \\&=& 
	\prod_{v \in D} RR_{v\mid x_U}, 
\end{eqnarray*}
that is, the $D$-product  relative risk $RR_{D\mid x_U}$   is  function of the causal  relative risks $RR_{v\mid x_U}$ for single nested outcomes, for any $v \in D$ and $D \subseteq V$.  
This represents an extreme case where $g[P(Y_{D}(1)\mid  X_U=x_U)]=g[P(Y_{D}(0)\mid  X_U=x_U)]$ because of independence.  In this case there is no joint effect and the causal effect is totally given by the marginal effect.


\subsection{Observed and Missing Potential Outcomes}
Unfortunately, we cannot directly observe both $Y_V(0)$ and  $Y_V(1)$ for any subject.
After the treatment has taken on a specific level, only the potential outcomes corresponding to that
level are realized and can be actually observed. Formally, let $W$ denote the actual treatment assignment: 
$W = 0$ for units assigned to the control group, and $W = 1$ for units assigned to the treatment group.
We observe $\Yobs_V \equiv Y_V(W) = W \cdot Y_V(1) +  (1-W) \cdot Y_V(0)$, but the other potential outcomes, 
$\Ymis_V \equiv  Y_V(1-W) = (1-W) \cdot Y_V(1) +  W \cdot Y_V(0)$, are missing. 
Therefore, causal inference problems under the  potential outcome approach are inherently missing data
problems, and some assumption on the treatment assignment mechanism is required to draw inference on causal effects.

In what follows, we will maintain the following assumption:	
\begin{assumption} \label{Ass:Random} Random treatment  assignment: 
	$$P\left(W  \mid Y_V(0), Y_V(1), X_U \right) = 	P(W)$$ 
\end{assumption}
Random assignment of the treatment, which usually holds by design in randomized experiments,
can be easily relaxed assuming that treatment assignment is independent of potential outcomes conditional on the observed covariates:  $P\left(W \mid Y_V(0), Y_V(1), X_U \right) = 	P(W \mid X_U)$.

Assumption \ref{Ass:Random} guarantees that the comparison of treated and control units  leads to valid inference on causal effects. 
In this paper we propose a model-based approach to causal inference deriving maximum likelihood estimators of the causal parameters of interest. 
Henceforth, we assume that $Y_V(0)$ and $Y_V(1)$ are independent, conditional on the covariates. This assumption has little inferential effect for causal estimands that do not depend on the association between individual potential outcomes as the super-population causal effects we focus on \cite[see, e.g.,][Chapter 8, for further details]{ImbensRubin:2015}.

\section{A Regression model for multiple binary potential outcomes} \label{sec:models}

\subsection{A multivariate model for multiple potential outcomes}
We assume that  the random vector $Y_V(w)\mid \{X_U=x_U\}$ with $x_U \in \I_U$ for  $w=0,1$, follows a multivariate Bernoulli distribution with probability parameter vector $\pi_{V\mid x_U}(w)=(\pi_{D\mid x_U}(w))_{D \subseteq V}$. The generic element $\pi_{D\mid x_U}(w)$, $w=0,1$,  is the following  joint probability:
$$
\pi_{D\mid x_U}(w)=P(Y_D(w)=1_D, Y_{V \setminus D}(w)=0_{V \setminus D} \mid X_U=x_U), \; D \subseteq V, \; x_U \in \I_U.
$$ 

Let $\mu_{V\mid x_U}(w)=(\mu_{D\mid x_U}(w))_{D \subseteq V}$ be the mean parameter, where the generic element, $\mu_{D\mid x_U}(w)$, is the marginal probability of the event $Y_D(w)=1_D$   given the covariate set $X_U=x_U$:
$$
\mu_{D\mid x_U}(w)=P(Y_D(w)=1_D\mid X_U=x_U), \qquad D \subseteq V, \;\; x_U \in \mathcal{I}_U,
$$ 
with $\mu_{\emptyset\mid x_U}(w)=1$, $w=0,1$.
It follows that the conditional distribution of 
a $D$-product potential outcome, $Y^D(w)\mid \{X_U=x_U\}$,  is an univariate Bernoulli random variable with probability parameter $\mu_{D\mid x_U}(w)$,  $D \subseteq V$ and $w=0,1$. Then, the causal relative risks in Equation \eqref{eq:D-product-causal-effect}  can be also written as function of the mean parameters, for any $x_U \in \mathcal{I}_U$:
\begin{equation}\label{eq:RR-mu}
RR_{D\mid x_U}=\frac{\mu_{D\mid x_U}(w=1)}{\mu_{D\mid x_U}(w=0)}, \qquad D \subseteq V.
\end{equation}

We now introduce the log-mean linear   parameterization developed by \citet{rov-lup-lar:2013}, which is the core of the   regression  framework we use for modeling multiple binary non-independent potential outcomes.

	Given the probability distribution of a  random vector $Y_V(w)\mid \{X_U=x_U\}$
	with  mean parameter $\mu_{V\mid x_U}(w)$, $w=0,1$, let $\gamma_{V\mid x_U}(w)=(\gamma_{D\mid x_U}(w))_{D \subseteq V}$ be the log-mean linear parameter vector with
	\begin{equation}\label{eq:LML}
	\gamma_{D\mid x_U}(w)=\sum_{D' \subseteq D}(-1)^{|D \setminus D'|} \log \mu_{D'\mid x_U}(w), \qquad D \subseteq V,\;\; x_U \in \I_U.
	\end{equation}	

The term  $\gamma_{D\mid x_U}(w)$, to which we refer as  log-mean linear interaction, represents a measure of association in the joint distribution of $Y_D(w)\mid \{X_U=x_U\}$, for any $D \subseteq V$. 

\subsection{Log-mean linear regression models}

\citet{lup-rov:2016} show that using the log-mean linear parameterization as link function for categorical response variables, the class of log-mean linear regression models results. We extend this approach for modeling the joint distribution of potential outcomes conditional on covariates, $Y_V(w)\mid \{X_U=x_U\}$, $w=0,1$ and $x_U \in \I_U$. Of course alternative model specifications can be considered, 
but we consider this method appealing because the model parameters directly  provide information on the causal estimands introduced in Section \ref{sec:causal}. Moreover the decomposition into marginal and joint effect can be easily expressed as function of   model parameters. 

%
	
Formally, the two saturated log-mean linear regression models for the conditional distribution of each potential outcome  $Y_V(w)\mid X_U$ for $w=0,1$ are given by

%

 
\begin{equation}\label{eq:LML-regression}
\begin{array}{ccl}
	\gamma_{D\mid x_U}(w=0)&\!\!=\!\!& \alpha_D +  \sum\limits_{E \subseteq  U} \alpha_{D\mid E}, \quad D \subseteq V  	\\
	\gamma_{D\mid x_U}(w=1)&\!\!=\!\!& \alpha_D + \alpha_D(w=1) + \sum\limits_{E \subseteq  U} \alpha_{D\mid E} + \sum\limits_{E\subseteq U}\alpha_{D\mid E}(w=1), \, D \subseteq V.\\
\end{array}
\end{equation}

The causal effect of the treatment on the log-mean linear interaction  can be defined as
\begin{equation}\label{eq:LML-causal effect}
\gamma_{D\mid x_U}(w=1)-\gamma_{D\mid x_U}(w=0)=\alpha_D(w=1) + \sum_{E\subseteq U}\alpha_{D\mid E}(w=1), \quad D \subseteq V.
\end{equation}
In particular, the parameter $\alpha_D(w=1)$ corresponds to the causal effect on the log-mean linear interaction  given the baseline level $x_U=0$, and the parameters $\alpha_{D\mid E}(w=1)$ represent the treatment effect heterogeneity for different covariate configurations, with $E \subseteq U$. For instance, the causal effect  on the log-mean linear interaction  for  $x_E=(1_E,0_{U\setminus E}) \in \mathcal{I}_U$ is 
\begin{eqnarray} \label{eq:LMLcausal-effect-xE}
\lefteqn{\gamma_{D\mid x_E}(w=1)-\gamma_{D\mid x_E}(w=0)= }\\&&
\alpha_D(w=1) + \sum_{E'\subseteq E}\alpha_{D\mid E'}(w=1),  \quad E \subseteq U,\; D \subseteq V. \nonumber
\end{eqnarray}

As far as the remaining parameters are concerned,  $\alpha_D$ represents the intercept and $\alpha_{D\mid E}$ corresponds to the effect of covariates $X_E$ on the log-mean linear interaction given $X_{U\setminus E}=0_{U\setminus E}$,  for any $E \subseteq U$ and $D \subseteq V$.

The following  lemma and theorem show that  we can  derive the causal relative risk  on each product outcome $Y^D$, $D \nesq V$,  defined in Equation \eqref{eq:D-product-causal-effect}, by combining  causal effects  on log-mean linear interactions in Equation \eqref{eq:LML-causal effect}.

\begin{lemma}\label{lemma:betaga-betamu}
	Under the log-mean linear regression models in Equation \eqref{eq:LML-regression}, for the baseline level $x_U=0$ of the covariate set $X_U$, we have
	\begin{eqnarray}\label{eq:LML-alpha}
	\alpha_D(w=1) &=&	    \sum_{D' \subseteq D}(-1)^{| D\setminus D'|}\log RR_{D'\mid x_U=0}, \qquad D \subseteq V.
	\end{eqnarray}
\end{lemma}

The following theorem shows how causal relative risks are given combining model parameters.
\begin{theorem}\label{prop:RR}
	Under the log-mean linear regression models in Equation \eqref{eq:LML-regression}, 
	for any product outcome $Y^D$, the  causal 	relative risk given the baseline level $x_U=0$ of the covariate set $X_U$ is 
	\begin{eqnarray}\label{eq:prop}
	RR_{D\mid x_U=0} &=& \exp \left\{\sum_{D' \subseteq D}\alpha_{D'}(w=1)\right\}, \qquad D \subseteq V.\label{eq:LML-RR}
	\end{eqnarray}
\end{theorem}

In the sequel, Equation \eqref{eq:LML-RR} is often written as $RR_{D\mid x_U=0} =  \prod_{D' \subseteq D} \exp\{$ $\alpha_{D'}(w=1)\}$,  $D \subseteq V$.
For the special case with $|D|=1$, the  causal effect for a single outcome given the  baseline level of the covariates is $RR_{D\mid x_U=0}= \exp \{\alpha_D(w=1)\}$.

The following corollary formally shows that the  causal relative risk for any sub-population defined by any covariate configuration $x_U \in \mathcal{I}_U$ is function of the log-mean linear regression coefficients. 
\begin{corollary}\label{cor:heterogenity}
	The  relative risk of a product outcome $Y^D$ for any value $x_U \in \I_U$  is 
	\begin{equation}\label{eq:cor}
	RR_{D\mid x_U}=\prod_{D' \subseteq D}\exp  \left\{\alpha_{D'}(w\!=\!1)+\sum_{E \subseteq U} \alpha_{D'\mid E}(w\!=\!1)\right\}, \quad D \subseteq V.
	\end{equation}
\end{corollary}	

Note that, if $\alpha_{D\mid E}(w=1)=0$, $D \subseteq V$ and $E \subseteq  U$, then   causal effects for $Y^D$ are homogeneous,  that is, the   relative risks for sub-populations defined by any $x_U \in I_U$ are all equals to the  relative risk in Equation \eqref{eq:LML-RR}.

\begin{example} \label{ex:ExLML}
In the morphine study with $Y_V(w)=(Y_1(w), Y_2(w))$, we observe two covariates: gender and age. 
We can dichotomize the variable age and construct a vector $X_U=(X_3,X_4)$ of two binary covariates, gender and age.
The log-mean linear regression model for $Y_V(0)\mid X_U$   is given by the following three equations
$$
\left. 
\begin{array}{lclclclcl}
\gamma_{1\mid 34}(w=0)  &=&  \alpha_1    &+&   \alpha_{1\mid 3}  &+&   \alpha_{1\mid 4} &+&  \alpha_{1\mid 34}   \\  
\gamma_{2\mid 34}(w=0)  &=&  \alpha_2    &+& \alpha_{2\mid 3}  &+&   \alpha_{2\mid 4} &+&  \alpha_{2\mid 34}  \\
\gamma_{12\mid 34}(w=0) &=& \alpha_{12}  &+&  \alpha_{12\mid 3}  &+&   \alpha_{12\mid 4} &+&  \alpha_{12\mid 34}
\end{array}%
\right.
$$
and the log-mean linear regression model for $Y_V(1)\mid X_U$ is given by the following three equations:
$$
\begin{array}{lc lc lc lc lc lc lc l}
\gamma_{1\mid 34}(w=1)  &\!\!=\!\!&  \alpha_1  &\!\!+\!\!& \alpha_1(w=1) &\!\!+\!\!&
\alpha_{1\mid 3} \,\,+\alpha_{1\mid 4}\,\,\,+ \alpha_{1\mid 34} \,\,+\\
&& \alpha_{1\mid 3}(w=1) &\!\!+\!\!& \alpha_{1\mid 4}(w=1)&\!\!+\!\!&\alpha_{1\mid 34}(w=1) \\
\gamma_{2\mid 34}(w=1) &\!\!=\!\!&   \alpha_2  &\!\!+\!\!& \alpha_2(w=1) &\!\!+\!\!&
\alpha_{2\mid 3} \,\,+\alpha_{2\mid 4}\,\,\,+ \alpha_{2\mid 34} \,\,+\\
 & & \alpha_{2\mid 3}(w=1) &\!\!+\!\!& \alpha_{2\mid 4}(w=1)&\!\!+\!\!&\alpha_{2\mid 34}(w=1) \\
\gamma_{12\mid 34}(w=1)  &\!\!=\!\!& 
\alpha_{12}&\!\!+\!\!& \alpha_{12}(w=1) &\!\!+\!\!& \alpha_{12\mid 3}+ \alpha_{12\mid 4}+\alpha_{12\mid 34}  +  \\
&&
\alpha_{12\mid 3}(w=1) &\!\!+\!\!& \alpha_{12\mid 4}(w=1) &\!\!+\!\!& \alpha_{12\mid 34}(w=1) 
\end{array}
$$ 
The parameters $\alpha_1(w=1)$ and $\alpha_2(w=1)$ are related to the  causal relative risks given the covariates' baseline  level for $Y_1$ and $Y_2$, respectively. Specifically $RR_{1\mid x_U=0}=\exp
\{\alpha_1(w=1)\}$ and $RR_{2\mid x_U=0} = \exp \{\alpha_2(w=1)\}$.  
Consider now the parameter $\alpha_{12}(w=1)$. We have
\begin{eqnarray*}
\alpha_{12}(w=1) &=& \sum_{D' \subseteq D}(-1)^{| D\setminus D'|}\log RR_{D'\mid x_U=0} \\
&=& \log RR_{12\mid x_U=0} -\log RR_{1\mid x_U=0} - \log RR_{2\mid x_U=0}.
\end{eqnarray*}
Therefore, $RR_{12\mid x_U=0}= \exp\{\alpha_{1}(w=1)+\alpha_{2}(w=1)+\alpha_{12}(w=1)\}$.

In this simple example with two binary covariates, for each binary outcome $Y_1$, $Y_2$ and $Y^{\{1,2\}}$, using Corollary \ref{cor:heterogenity}, we can derive four  causal relative risks, depending on the value of the two binary covariates, $X_3$ and $X_4$. For instance, given $x_U=\{x_3=1,x_4=0\}$,
\begin{eqnarray*}
RR_{1\mid x_U} &=& \exp \left\{ \alpha_1(w=1)+ \alpha_{1\mid 3}(w=1) \right\},\\
RR_{2\mid x_U} &=& \exp \left\{ \alpha_2(w=1)+ \alpha_{2\mid 3}(w=1) \right\},\\
RR_{12\mid x_U}&=& \exp \left\{ \alpha_{12}(w=1)+ \alpha_{12\mid 3}(w=1) \right\} \times RR_{1\mid x_U} \times RR_{2\mid x_U}.
\end{eqnarray*}
\end{example}

\subsection{Marginal and joint causal effects using log-mean linear parameters}
We now show how the decomposition of the causal relative risk on a product outcome into its marginal and joint components  naturally follows with a straightforward interpretation, using the log-mean linear regression model in Equation \eqref{eq:LML-regression}.  In particular, we specify the functions in Equations  \eqref{eq:ICE1} and \eqref{eq:ECE1}  in terms  of model parameters.
For simplicity, we  focus on  the baseline level of the covariates, $x_U=0$, but the following reasoning  applies to every level of the covariates $x_U \in \I_U$.
	Then,  for the baseline level $x_U=0$ of the covariate set $X_U$, respectively, as follows: 
	\begin{eqnarray}
	MCE_{D\mid x_U=0}&=& \prod_{D' \subset D}\exp \{\alpha_{D'}(w=1)\}, \qquad D \subseteq V,\label{eq:ICE}\\
	JCE_{D\mid x_U=0} &=& \exp \{\alpha_{D}(w=1)\}, \qquad D \subseteq V.\label{eq:ECE}
	\end{eqnarray}
%
From Equation \eqref{eq:LML-causal effect}, we have that Equation \eqref{eq:ECE} is function of the joint probability of $Y_D(w)\mid \{X_U=0\}$, for $w=0,1$ in accordance with Equation \eqref{eq:ECE1}; in particular, from Lemma \ref{lemma:betaga-betamu} we have that the adopted measure of association is $g[P(Y_D(w)\mid X_U=0)]= \sum_{D' \subseteq D}(-1)^{|D\setminus D'|}\log P[Y_{D'}(w)=1\mid X_U=0]$ for $w=0,1$, and that
\begin{equation}
	JCE_{D\mid x_U=0}=g[P(Y_D(1)\mid X_U=0)]/g[P(Y_D(0)\mid X_U=0)], \quad D \subseteq V.
\end{equation} 
Therefore, $JCE_{D\mid x_U=0}=1$ if there is  no joint effect.   The following lemma shows that  Equation \eqref{eq:ICE} is function of the causal estimands we propose, that is, the causal relative 
risks for nested outcomes,  $Y^{D'}$, with $D' \nes D$, following the general definition in Equation \eqref{eq:ICE1}.
\begin{lemma}\label{lem:ICE}
	Under the log-mean linear regression models in Equation \eqref{eq:LML-regression},  for any $Y^D$, 
	\begin{eqnarray}\label{eq:RRcombination}
	MCE_{D\mid x_U=0}=\Bigg[ \prod_{D' \subset D} RR_{D'\mid x_U=0}^{(-1)^{| D\setminus D'| }} \Bigg]^{-1}, \qquad D \subseteq V.
	\end{eqnarray}
\end{lemma}
Equations \eqref{eq:ICE} and \eqref{eq:ECE}   can be extended to every value  of the covariate set $X_U$. Specifically, for every $D \subseteq V$ and $x_U \in \I_U$,
\begin{eqnarray}\label{eq:ICE-xU}
MCE_{D\mid x_U}&=& \prod_{D' \subset D}\exp \left\{\alpha_{D'}(w=1)+\sum_{E \subseteq U} \alpha_{D'\mid E}(w=1)\right\},\\
JCE_{D\mid x_U} &=& \exp \left\{\alpha_{D}(w=1)+ \sum_{E \subseteq U}\alpha_{D\mid E}(w=1)\right\}.\label{eq:ECE-xU}
\end{eqnarray}
Lemma \ref{lem:ICE} can be also generalized for each value of the covariate set $X_U$  using Corollary \ref{cor:heterogenity}.

The  proposition below shows that for every product outcome $Y^D \in Y^{\cal{V}}$, the decomposition of the  causal effect into its marginal and joint components  naturally follows.
\begin{proposition}\label{pro:decomposition}
	Under the log-mean linear regression models in Equation \eqref{eq:LML-regression}, for any product outcome $Y^D$,
	\begin{eqnarray}\label{eq:ICE+ECE}
	RR_{D\mid x_U} &=& JCE_{D\mid x_U} \times MCE_{D\mid x_U}, \qquad D \subseteq V,
	\end{eqnarray}
	given any  $x_U \in \I_U$ of the covariate set $X_U$.
\end{proposition}
The case $|D|=1$ is trivial  because $ JCE_{D\mid x_U}=RR_{\emptyset\mid x_U}=1$ and $ RR_{D\mid x_U}= MCE_{D\mid x_U}$.

\begin{corollary}\label{cor:ECE0}
	Under the log-mean linear regression models in Equation \eqref{eq:LML-regression}, for any product outcome $Y^D$, and for any $x_U \in \mathcal{I}_U$,  
	\begin{equation}\label{eq:zero-alpha}
	RR_{D\mid x_U}= MCE_{D\mid x_U}, \qquad D \subseteq V
	\end{equation}
	if and only if  
	\begin{equation}\label{eq:cor3}
	\alpha_{D}(w=1) + \sum_{E \subseteq U}\alpha_{D\mid E}(w=1) = 0, \qquad D \subseteq V
	\end{equation}

\end{corollary}
A special case of Corollary \ref{cor:ECE0}, is for baseline level $x_U=0$ of the covariates, when $RR_{D\mid x_U=0}= MCE_{D\mid x_U=0}$ if and only if $\alpha_{D}(w=1)=0$, for any $D \subseteq V$.
The following corollary shows that  Corollary \ref{cor:ECE0} necessary holds in case of independence.
\begin{corollary}\label{cor:ECE-indip}
	For any potential outcome $Y_D(w)$ with $D \subseteq V$ and $w=0,1$, suppose that there exist a partition $A$ and $B$ of $D$ with $A,B \neq \emptyset$ and $A \cap B=\emptyset$ such that $Y_A(w) $ and $Y_B(w)$
	are independent given  $X_U=x_U$, for any value $x_U \in \mathcal{I}_U$ of the covariates $X_U$. Then, under the log-mean linear regression model in Equation \eqref{eq:LML-regression}, for any product outcome $Y^D$, and for any $x_U \in \mathcal{I}_U$,
	\begin{equation}\label{eq:zero-alpha-indip}
	RR_{D\mid x_U}= MCE_{D\mid x_U}, \qquad D \subseteq V.
	\end{equation}
\end{corollary}

Following Corollary \eqref{cor:ECE-indip}, we remark that independence is a sufficient but not a necessary condition in order to have $RR_{D\mid x_U}=MCE_{D\mid x_U}$, for any $D \subseteq V$ and $x_U \in \mathcal{I}_U$. In particular, Equation \eqref{eq:zero-alpha} holds whenever the treatment  affects  the product  outcome $Y^D$, but has no effect on the association structure of the joint distribution of $Y_D$, $D \subseteq V$. Then, it is possible that  causal effects on marginal outcomes contain all relevant information about the causal effects of the treatment,  and thus the association among outcomes can be neglected,  even in case of non-independent outcomes.

To be thorough, it is worth mentioning another, although less interesting, case, too.  
If $\alpha_{D'}(w=1)=0$ for every $D' \subset D$, for a given $D \subseteq V$ with $|D|>1$, then $MCE_{D\mid x_U=0}=1$, and we get $ RR_{D\mid x_U=0}=  JCE_{D\mid x_U=0}$. Therefore, it could be possible that the treatment has a casual effect on $Y^D$, even if it has no effect on every nested product outcome, $Y^{D'}$, $D' \subset D $. In this case, the casual effect on $Y^D$ would be totally given by the effect of the treatment on the dependence structure of  $Y_D$. 
\begin{example}
	Consider the scenario described in Example \ref{ex:ExLML}, where we have $Y_V(w)=(Y_1(w), Y_2(w))$ and  a vector $X_U=(X_3,X_4)$ of two binary covariates, and the log-mean linear regression model for $Y_V(w)\mid X_U$, $w=0,1$.
	Suppose we are interested in the  relative risk $RR_{\{1,2\}\mid x_U=0}$. 
	This effect can be decomposed into the  marginal causal effect and   the  joint causal effect, which can be written, using the log-mean linear  model coefficients, as follows:
	\begin{eqnarray*}
	MCE_{12\mid x_U=0}&=& \exp \{\alpha_1(w=1)\} \times \exp \{\alpha_2(w=1)\} \\ 
	JCE_{12\mid x_U=0} &=& \exp \{\alpha_{12}(w=1)\}.
	\end{eqnarray*}
	Then, 
	$$
	RR_{12\mid x_U=0}= \exp \{\alpha(w=1)\} \times \exp \{\alpha_2(w=1)\} \times \exp \{\alpha_{12}(w=1)\}.
	$$
	If the logarithm of the joint effect is null, i.e. $\alpha_{12}(w=1)=0$, then causal effects on $Y^{\{1,2\}}$ are only through  causal effects on each marginal outcome, $Y_{1}$ and $Y_{2}$. 
	The same result also holds for every level $x_U \in \I_U$ of the covariates. For instance, consider $x_U=\{x_3=1,x_4=0\}$, then 
	$
	MCE_{12\mid x_U}= \exp \{\alpha_1(w=1)+ \alpha_{1\mid 3}(w=1)\} \times \exp \{\alpha_2(w=1) + \alpha_{2\mid 3}(w=1)\}
	$
	and
	$
	JCE_{12\mid x_U}= \exp\{\alpha_{12}(w=1)+ \alpha_{12\mid 3}(w=1)\}.
	$
	In the log-mean linear regression model, which involves hierarchical effects, $\alpha_{12}(w=1)=0$ implies that $\alpha_{12\mid E}(w=1)=0$ for any $E \subseteq U$, therefore if $\alpha_{12}(w=1)=0$ then   $ JCE_{12\mid x_U}=1$ and $ RR_{12\mid x_U}= MCE_{12\mid x_U}$.
\end{example}

\subsection{Inference} 
For inference, we use a maximum likelihood approach. In particular, maximum likelihood estimators  are obtained by implementing an algorithm inspired on the maximization procedure developed in \citet{lang:1996}, properly adjusted for working out the estimates of the causal effects of interest and the corresponding standard errors. A similar maximization procedure has been  also discussed by \citet{lupparelli:2006} in the context of marginal models. For technical details and a review of further maximization approaches see also \citet{eva-for:2013} and references therein.

\section{Applications} \label{sec:applications}
\subsection{The Morphine Study}
The morphine study is a prospective, randomized, double-blind study concerning the effect of preoperative oral administration of morphine sulphate on postoperative pain relief. The study involved a random sample of $n = 60$ patients aged $18 - 80$ who were undergoing an elective open colorectal abdominal surgery. 
Out of these 60 patients,  32  were randomly assigned to the treatment group, and 28  were randomly assigned to the control group.  Let $W$ denote the observed treatment variable. Before surgery, patients in the treatment group with $W=1$, were administered oral morphine sulphate (Oramorph$^{\textregistered}$, Molteni Farmaceutici, Italy), and patients in the control group with $W=0$, received oral midazolam (Hypnovel$^{\textregistered}$, Roche, Switzerland), a short-acting drug inducing sedation, which is considered as an active placebo. 

The outcome of primary interest is post-operative pain intensity measured using visual analogue scale
 scores at rest and for movement (that is, upon coughing).  Visual analogue scale scores are measured using a line of 100 mm where the left extremity is no pain and the right one is extreme pain. 
Here we focus on pain intensity at rest and for movement 4 hours after the end of surgery \cite[see][for further details on the study]{Borracci2013}.
Physicians consider a pain score not greater than 30 mm at rest, and not greater than 45 mm on movement as a satisfactory pain relief. Therefore we dichotomize the two outcome variables using 30 and 45 as cutoff points for 
pain intensity at rest and for movement, respectively.   Let $Y_1=Y_S$ and $Y_2=Y_{Dy}$ denote the binary indicators for low  versus high visual analogue scale scores  at rest and for movement, respectively.

Under SUTVA,   let $Y_S(w)$ and $Y_{Dy}(w)$  define the potential outcomes for pain intensity at rest and for movement, respectively,  given assignment to treatment level $w$: $Y_S(w)$ and $Y_{Dy}(w)$ are binary variables equal to 1 for patients with visual analogue scale score at rest and for movement not greater that 30 mm and 45 mm, respectively, given assignment $w$, and 0 otherwise. Let $Y_S(W)$ and $Y_{Dy}(W)$ be the actual outcomes observed. 
For each patient we also observe two covariates, gender, $X_G$  ($X_G=0$ for females; and $X_G=1$ for males),  and age in years. We dichotomize the variable age considering a binary variable $X_A$ equal to 1 for patients older than 65 years, and $0$ otherwise. So we get the vector $X_U=(X_A,X_G)$ of two binary covariates.

\begin{table}
	\centering \caption{Maximum likelihood estimates of the log-mean linear regression model for $\{Y_S(w),Y_{Dy}(w)\}\mid \{X_A,X_G\}$ (in brackets the standard errors).}\label{tab.2outcomes}
	\begin{tabular}{lrrrr}
		\hline \vspace{-0.25cm}\\ 
		$Y^D(w)\mid X_U$ &  \multicolumn{1}{c}{$\hat{\alpha}_D$} &  \multicolumn{1}{c}{$\hat{\alpha}_D(w=1)$} & \multicolumn{1}{c}{$\hat{\alpha}_{D\mid A}$} & \multicolumn{1}{c}{$\hat{\alpha}_{D\mid G}$}\\ 
		\hline  \vspace{-0.25cm}\\ 
		$Y^{S}(w)\mid X_U$ & -1.600 (0.343) & 1.008 (0.325) & 0.205 (0.209) & 0.270 (0.195)\\
		$Y^{Dy}(w)\mid X_U$ & -2.437 (0.560) & 1.122 (0.498)  & 0.700 (0.439) & 0.024 (0.330)\\
		$Y^{\{S,Dy\}}(w)\mid X_U$ & 1.148 (0.399)  & -0.592 (0.381) & -0.178 (0.204) & -0.269 (0.194)\\  
		
		\hline 
	\end{tabular}
\end{table}
Let us consider the log-mean linear regression models for  $\{Y_S(0),Y_{Dy}(0)\}\mid \{X_A,X_G\}$ and 
$\{Y_S(1),Y_{Dy}(1)\}\mid \{X_A,X_G\}$  with no-interaction terms: $\alpha_{D\mid A,G}$ $=0$, and $\alpha_{D\mid E}(w=1)=0$, for each $D  \subseteq \{S,Dy\}$ and $E\in \{A,G\}$:
\begin{eqnarray*}
\gamma_{S\mid A,G}(w=0)   &=& \alpha_S      +  \alpha_{S\mid A} + \alpha_{S\mid G}, \\
\gamma_{Dy\mid A,G}(w=0)  &=& \alpha_{Dy}   + \alpha_{Dy\mid A} + \alpha_{Dy\mid G},\\
\gamma_{S,Dy\mid A,G}(w=0)&=& \alpha_{S,Dy} + \alpha_{S,Dy\mid A} + \alpha_{S,Dy\mid G}
\end{eqnarray*}%
and
\begin{eqnarray*}
\gamma_{S\mid A,G}(w=1) &=& \alpha_S +   \alpha_S(w=1) + \alpha_{S\mid A} +\alpha_{S\mid G}, \\  
\gamma_{Dy\mid A,G}(w=1) &=& \alpha_{Dy} + \alpha_{Dy}(w=1)  + \alpha_{Dy\mid A} + \alpha_{Dy\mid G},\\
\gamma_{S,Dy\mid A,G}(w=1)&=& \alpha_{S,Dy}+ \alpha_{S,Dy}(w=1)    + \alpha_{S,Dy\mid A} + \alpha_{S,Dy\mid G}
\end{eqnarray*}%
This model specification implies that treatment effects are homogeneous across sub-populations defined by the values of the two covariates. Therefore, in the following, we do not need to specify the conditioning set, e.g. $x_U=0$, for any causal estimand. 

The model has 12 degrees of freedom and gives a good fitting with deviance $11.753$ ($p$-value$= 0.466$),  and $BIC= 332.807$. 
The estimates in Table \ref{tab.2outcomes} show a positive causal effect of treatment for both single outcomes. 
%
However, the joint causal effect is not statistically significant, suggesting that the treatment has not effect on the dependence structure between pain intensity at rest and on movement. The effect of the two covariates is not significant, too. Therefore we repeat the analysis setting  $\alpha_{D\mid A}=\alpha_{D\mid G}=0$ for any $D \subseteq \{S,Dy\}$ and with $\alpha_{\{S,Dy\}}(w=1)=0$. This model has 19 degrees of freedom and gives a good fitting with deviance $18.775$, $p$-value $= 0.471$, and $BIC= 311.168$. Notice that the same model with $\alpha_{\{S,Dy\}}(w=1) \neq 0$ still gives a not significant estimate for this parameter, i.e., $\hat{\alpha}_{\{S,Dy\}}(w=1)=-0.617$ $(se=0.399)$.
The estimates collected in Table \ref{tab.2outcomes_reduced} still show a positive causal effect of treatment on both single outcomes, i.e., $\hat{RR}_{S\mid x_U}= \exp(0.987)=2.683$, and $\hat{RR}_{Dy\mid x_U}= \exp(1.279)=3.593$. For the product outcome $Y^{\{S,Dy\}}$, the total causal effect coincides with the marginal causal effect, because   the logarithm of the joint causal effect is assumed to be  null: $\log(JCE_{\{S,Dy\}\mid x_U})= 0$, and thus $
\hat{RR}_{\{S,Dy\}\mid x_U}= \hat{MCE}_{\{S,Dy\}\mid x_U}= 2.683 \times 3.593=9.640.$
\begin{table}
	\centering\caption{Maximum likelihood estimates of the log-mean linear regression model for $\{Y_S(w),Y_{Dy}(w)\}\mid \{X_A,X_G\}$ with zero constraints denoted by - (in brackets the standard errors).}\label{tab.2outcomes_reduced}
	\begin{tabular}{l|rrrr}\hline
		\vspace{-0.25cm}\\
		$Y^D(w)\mid X_U$ &  \multicolumn{1}{c}{$\hat{\alpha}_D$} &  \multicolumn{1}{c}{$\hat{\alpha}_D(w=1)$} & \multicolumn{1}{c}{$\hat{\alpha}_{D\mid A}$} & \multicolumn{1}{c}{$\hat{\alpha}_{D\mid G}$}\\
		\hline	\vspace{-0.25cm}\\
		$Y^{S}(w)\mid X_U$ & -1.310 (0.297) & 0.987 (0.313) & - & -\\
		$Y^{Dy}(w)\mid X_U$ & -2.054 (0.459) & 1.279 (0.494)  & - & -\\
		$Y^{\{S,Dy\}}(w)\mid X_U$ & 0.302 (0.386)  & - & - & -\\ \hline 
	\end{tabular}
\end{table}

\subsection{The Honey data study}
We consider a double-blinded randomized study taken by  \citet{paul-al:2007}, 
where the focus is on evaluating  the effects of a single nocturnal dose of buckwheat honey or honey-flavored dextromethorphan  versus no treatment on nocturnal cough and sleep difficulty associated with childhood upper respiratory tract infections.  A sample of patients aged between 2 and 18 years with cough attributes characterized by the presence of rhinorrhea and cough for 7 or fewer days duration have been enrolled. Subjective parental assessments about their child cough symptoms  were assessed both previous and after the treatment administration   \cite[see][for details]{paul-al:2007}. The primary outcomes of interest are cough bothersome (`How bothersome was your child's coughing last night?'), cough frequency (`How frequent was your child's cough last night?') and cough severity  (`How severe was your child's cough last night?'). These outcomes are measured on a 7-point Likert scale from 0 (\lq not at all') to 6 (\lq extremely').
Using univariate statistical techniques, \citet{paul-al:2007}  found that honey may be a preferable treatment. Significant differences in symptom improvement were not detected between dextromethorphan and no treatment or between dextromethorphan and honey for any outcome of interest.

We apply the model-based approach described in Section \ref{sec:models}  for drawing inference on the causal effect of honey on children respiratory infection due to considering different combinations of three selected cough attributes: bothersome, frequency and severity. We dichotomize these variables merging levels 0 to 2 in level 1; and levels 3 to 6 in level 0.
We focus on the sub-sample of 72 children receiving honey or no treatment (ignoring children receiving dextromethorphan): 35 children were randomly assigned to the honey treatment $(W=1)$ and 37 children were randomly assigned to no treatment  $(W=0)$.    Then, under SUTVA, for each patient we get a vector $Y_V(w)=(Y_B(w),Y_F(w),Y_S(w))$ of three potential outcomes given assignment $w$, $w=0,1$. The variables $Y_B(w)$, $Y_F(w)$, and $Y_S(w)$ respectively take  value 1 in case of absent or low bothersome, cough frequency and severity, and 0 otherwise. 
We are also interested on the causal effect of honey on the product  outcomes $Y^D$ included in the augmented vector of $\YV$ for every $D \subseteq V$.  Each product outcome represents a combinations of different cough symptoms, which may jointly occur to children, e.g., $Y^{B,F}(w)=1$  for children with absent or low cough bothersome and frequency under treatment $w$, $w=0,1$.

\begin{table}[t]
	\centering\caption{Maximum likelihood estimates of the  log-mean linear regression model for $\{Y_F(w),Y_S(w),Y_{B}(w)\}\mid X$ (in brackets the standard errors).}\label{tab.honey}
	\begin{tabular}{l|rrr}
		\hline
		\vspace{-0.25cm}\\
		$Y^D(w)\mid X$ &  \multicolumn{1}{c}{$\hat{\alpha}_D$} &  \multicolumn{1}{c}{$\hat{\alpha}_D(w=1)$} & \multicolumn{1}{c}{$\hat{\alpha}_{D\mid X}$}\\
		\hline	\vspace{-0.25cm}\\
		$Y^B(w)\mid X$ & -1.050 (0.212) & 0.599 (0.205)  &  0.289 (0.170)\\
		$Y^F(w)\mid X$ & -1.468 (0.265) & 0.754 (0.240) &  0.602 (0.209)\\
		$Y^S(w)\mid X$ & -1.386 (0.254) & 0.532 (0.241) &  0.601 (0.209) \\
		$Y^{\{B,F\}}(w)\mid X$ & 0.875 (0.201)  & -0.502 (0.192)  & -0.270 (0.155)\\
		$Y^{\{B,S\}}(w)\mid X$ & 0.953 (0.201)  & -0.517 (0.192)  & -0.341 (0.161)\\
		$Y^{\{F,S\}}(w)\mid X$ & 1.211 (0.235) & -0.585 (0.211) & -0.523 (0.195)\\
		$Y^{\{B,F,S\}}(w)\mid X$ & -0.792 (0.195)  & 0.423 (0.185) &  0.269 (0.154)\\\hline
	\end{tabular}
\end{table}

We also consider an individual covariate $X$ obtained combining pre-treatment knowledge. Following \citet{paul-al:2007}, we built up a variable by summing the individual scores observed before the treatment about cough frequency, severity and bothersome such that we get a pre-treatment discrete indicator ranging from 0 to 18, which 
we dichotomize with respect to its (sample) median equal to 12.  The resulting binary variable $X$ takes on level 1 for values lower than the median and 0 otherwise.

We specified a log-mean linear regression model for $Y_V(w)\mid X$ assuming no treatment effect heterogeneity across sub-populations defined by the pre-treatment covariate $X$: $\alpha_{D\mid X}(w=1)=0$, for each $D \subseteq V$. This model implies that causal effects are homogeneous  between children with a high  pre-treatment health score and   children with a low  pre-treatment health score. 

The model shows a good fitting with 7 degree of freedom, deviance $10.093$ and $p$-value$=0.183$.  Parameter estimates are collected in Table \ref{tab.honey} and estimates of the causal effects are shown in Table \ref{tab.honey.effects}. 
We get  positive estimates of the honey causal effects on each single outcome,  but the  strongest effect is on reducing the cough frequency.
Treatment has also  a positive effect on product outcomes, improving  conditions of children suffering from  combinations of different cough symptoms. In particular the treatment appears to be more effective when the symptoms includes cough frequency.  All  joint effects are statistically significant, suggesting that the treatment has an effect on the association structure between the outcomes. The  strongest joint effect is found for  the product outcome $Y^{\{B,F,S\}}$, which combines critical cough frequency, severity and bothersome.
The causal relative risk on the product-potential outcomes $Y^D$ with $D \subseteq V$ are derived   combining   joint  and   marginal effects: for instance, $\hat{RR}_{\{B,F\}\mid x}= 0.606 \times 3.885=2.354$. 

Unlike the application on  the morphine data,  in this case we find strong evidence that the analysis cannot be conducted separately on each single outcome because for each product outcome $Y^D$ with $D \subseteq V$, the treatment effect on the association among single  outcomes $Y_v$ with $v \in D$ cannot be ignored.  Then, a multivariate approach is definitively more suitable than an univariate one.

\begin{table}[t]
		\centering \caption{Honey Study: Estimates of the causal effects (in brackets the standard errors)}\label{tab.honey.effects}
		\begin{center}
			\begin{tabular}{lc c l c}\hline
			Estimand       &  Estimate   & &Estimand       &  Estimate\\
			\hline
			\multicolumn{5}{c}{Causal Relative Risks}\\
			$RR_{B\mid x}$ & 1.820 (0.372) & & $RR_{\{B,F\}\mid x}$   & 2.341 (0.614)\\
			$RR_{F\mid x}$ & 2.125 (0.511) & & $RR_{\{B,S\}\mid x}$   & 1.847 (0.475) \\
			$RR_{S\mid x}$ & 1.702 (0.410) & & $RR_{\{F,S\}\mid x}$   & 2.015 (0.552) \\
			&               & &$RR_{\{B,F,S\}\mid x}$  & 2.022 (0.561) \\
			\\
			\multicolumn{2}{l}{ Joint Causal Effects } &&\multicolumn{2}{l}{Marginal Causal Effects}\\        
			$JCE_{\{B,F\}\mid x}$   & 0.605 (0.116) & &  $MCE_{\{B,F\}\mid x}$ & 3.869 (1.584)  \\
			$JCE_{\{B,S\}\mid x}$   & 0.596 (0.115) & &  $MCE_{\{B,S\}\mid x}$ & 3.099 (1.282) \\
			$JCE_{\{F,S\}\mid x}$   & 0.557 (0.118) & &  $MCE_{\{F,S\}\mid x}$ & 3.618 (1.630)\\
			$JCE_{\{B,F,S\}\mid x}$ & 1.527 (0.283) & &  $MCE_{\{B,F,S\}\mid x}$ & 1.324 (0.326)  \\   
			\vspace{-0.3cm}\\
			\hline
		\end{tabular}
	\end{center}
\end{table}

\section{Conclusion}\label{sec:conclusion}
Causal inference in the presence of multiple non-independent outcomes represents a challenging task for several reasons. In particular, an augmented set of outcomes needs to be considered because also the joint occurrence of combinations of outcomes becomes of interest. These \lq\lq new quantities\rq\rq\ need to be formalized together with an enlarged set of causal estimands, including the effect of treatment on combinations of outcomes.
We formalize  these concepts for binary  outcome variables by introducing the notion of product outcomes and by decomposing the treatment effect on these outcomes into the joint and marginal causal effects.  
A general definition for the marginal and the joint causal estimands   has been introduced,
although their specification necessarily requires to make some decisions, for instance, by introducing  modeling assumptions.

We propose to model the joint distribution of potential outcomes  using the class of log-mean linear regression models proposed by \citet{lup-rov:2016}.
Interestingly, the parameters of the resulting model are directly related to  the causal estimands of interest, and the analytic decomposition into marginal and joint effect naturally arises. 

Further  approaches for multiple binary responses may be also explored, such as the multivariate logistic regression of \citet{glo-mcc:1995}.
Nevertheless, we deem that the class of log-mean linear regressions is particularly appealing when the causal estimand of interest are causal relative risks. 

\appendix

\section{Appendix: Proofs}\label{app}

\begin{proof}[Proof of Lemma \ref{lemma:betaga-betamu}]
	Let us consider the causal effect in Equation \eqref{eq:LML-causal effect} for the baseline level $x_U=0$ of the covariates. For any $D \subseteq V$,
	$$
	\alpha_D(w=1)= \gamma_{D\mid x_U=0}(w=1)-\gamma_{D\mid x_U=0}(w=0).
	$$
	Then, from Equation \eqref{ eq:LML} we have that, for any $D \subseteq V$,
	$$
	\alpha_D(w=1)= \sum_{D' \subseteq D}(-1)^{|D \setminus D'|} \log \mu_{D'\mid x_U=0}(1)-\sum_{D' \subseteq D}(-1)^{|D \setminus D'|} \log \mu_{D'\mid x_U=0}(0).
	$$
	It follows that 
	$$
	\alpha_D(w=1) =	    \sum_{D' \subseteq D}(-1)^{|D\setminus D'|}\log RR_{D'\mid x_U=0}, \qquad D \subseteq V
	$$
	because 
	$\log RR_{D'\mid x_U=0}=\log\mu_{D'\mid x_U=0}(1)- \log \mu_{D'\mid x_U=0}(0)$.
\end{proof}

\begin{proof}[Proof of Theorem \ref{prop:RR}]
	First of all we remark that, given any set $D$, the power set $\mathcal{P}(D)=\{D'\}$ includes the same number $2^{|D|-1}$  of even and odd subsets, i.e., $\sum_{D' \subseteq D}(-1)^{|D\setminus D'|}=0$. 
	
	From Lemma \ref{lemma:betaga-betamu} we have that  
	\begin{equation}\label{eq:prop-bis}
	\sum_{D' \subseteq D}\alpha_{D'}(w=1)= \sum_{D' \subseteq D}\left\{\sum_{\tilde{D} \subseteq D'}(-1)^{|D'\setminus \tilde{D}|}\log RR_{\tilde{D}\mid x_U=0}\right\}, \qquad D \subseteq V.
	\end{equation}
	Equation \eqref{eq:prop-bis}  is equivalent to
	\begin{equation}\label{eq:prop-ter}
	\log RR_{D\mid x_U=0}+ \sum_{D' \subseteq D}\left\{\sum_{\tilde{D} \subseteq D':\tilde{D} \neq D}(-1)^{|D'\setminus \tilde{D}|}\log RR_{\tilde{D}\mid x_U=0}\right\},   D \subseteq V.
	\end{equation}
	The second addend in Equation \eqref{eq:prop-ter} is null if  for every $\tilde{D} \subset D$, among all supersets $D' \subseteq D$  such that $D' \supseteq \tilde{D}$, there is the same number of even and odd subsets $D' \setminus \tilde{D}$. 
	Notice that the case $\tilde{D} = \emptyset$ is trivial because $\log RR_{\emptyset\mid x_U=0}=0$ by definition. 
	
	Now consider any non-empty subset $\tilde{D} \subset D$. This is included in all supersets $D' \subseteq D$ of type $D'=\tilde{D} \cup K$ for any $K \in \mathcal{P}(D'\setminus \tilde{D})$ which is a power set having the same number of even and odd subsets. Then, the  result follows because
	\begin{equation}
	\sum_{D' \subseteq D}\left\{\sum_{\tilde{D} \subseteq D':\tilde{D} \neq D}(-1)^{|D'\setminus \tilde{D}|}\log RR_{\tilde{D}\mid x_U=0}\right\} =0, \qquad D \subseteq V.
	\end{equation}
	\hspace{1cm}
\end{proof}

\begin{proof}[Proof of Corollary \ref{cor:heterogenity}]
	From Equation \eqref{ eq:LML-causal effect} we have
	\begin{equation}\label{eq:cor1}
	\alpha_{D'}(w=1)+ \sum_{E \subseteq U} \alpha_{D'\mid E}(w=1)=\gamma_{D'\mid x_U}(w=1)-\gamma_{D'\mid x_U}(w=0), \quad D' \subseteq D,
	\end{equation}
	for every $D \subseteq V$ and $x_U \in \mathcal{I}_U$. From Equation \eqref{ eq:LML}, Equation \eqref{eq:cor1} is equal to
	\begin{equation}
	\sum_{\tilde{D}\subseteq D'} (-1)^{|D'\setminus \tilde{D}|} \log \mu_{\tilde{D}\mid x_U}(w=1)- \sum_{\tilde{D} \subseteq D'} (-1)^{|D'\setminus \tilde{D}|} \log \mu_{\tilde{D}\mid x_U}(w=0)
	\end{equation}
	which, by Equation \eqref{ eq:RR-mu}, is also equal to $\sum_{\tilde{D} \subseteq D'} (-1)^{|D' \setminus \tilde{D}|} \log RR_{\tilde{D}\mid x_U}$, for any $x_U \in \mathcal{I}_U$.
	Therefore, we have 
	\begin{eqnarray}
	\lefteqn{\sum_{D' \subseteq D}  \left\{\alpha_{D'}(w\!=\!1)+  \sum_{E \subseteq U} \alpha_{D'\mid E}(w\!=\!1)\right\}  = }\\&&
	\sum_{D' \subseteq D}\left\{\sum_{\tilde{D} \subseteq D'}(-1)^{|D'\setminus \tilde{D}|} \log RR_{\tilde{D}\mid x_U}\right\}, \quad D \subseteq V, \nonumber
	\end{eqnarray}
	for any $x_U \in \mathcal{I}_U$. The result follows applying the  proof of Theorem \ref{prop:RR}.
\end{proof}	

\begin{proof}[Proof of Lemma \ref{lem:ICE}]
	From Equation \eqref{eq:LML-alpha} and \eqref{eq:LML-RR} we have, respectively, 
	$$
	\alpha_D(w=1)= \log RR_{D\mid x_U=0} + \sum_{D' \subset D}(-1)^{|D\setminus D'|} \log RR_{D'\mid x_U=0}
	$$
	and
	$$
	\log RR_{D\mid x_U=0} = \alpha_D(w=1) + \sum_{D' \subset D} \alpha_{D'}(w=1),
	$$
	for any $D \subseteq V$. Then, the result follows given that, for any $D \subseteq V$,
	$$\sum_{D' \subset D} \alpha_{D'}(w=1)=-\sum_{D' \subset D}(-1)^{|D\setminus D'|} RR_{D'\mid x_U=0}.
	$$
	\hspace{1cm}
\end{proof}

\begin{proof}[Proof of Proposition \ref{pro:decomposition}]
	For every $D \subseteq V$ and any  $x_U \in \mathcal{I}_U$, the result follows by applying Corollary \ref{cor:heterogenity} to the product of $JCE_{D\mid x_U}$ and $MCE_{D\mid x_U}$, as respectively defined in Equations \eqref{eq:ICE-xU} and \eqref{eq:ECE-xU}.
\end{proof}

\begin{proof}[Proof of Corollary \ref{cor:ECE0}]
	Suppose that Equation \eqref{eq:zero-alpha} holds. Then  Proposition \ref{pro:decomposition} implies that 
	$JCE_{D\mid x_U}=1$ for any $D \subseteq V$, and thus Equation \eqref{ eq:cor3} holds given the definition of 
	$JCE_{D\mid x_U}$ in Equation \eqref{eq:ECE-xU}.
	Vice versa, if Equation \eqref{eq:cor3} holds, then  $JCE_{D\mid x_U}=1$ for any $D \subseteq V$, given the definition of 
	$JCE_{D\mid x_U}$ in Equation \eqref{eq:ECE-xU}. Equation \eqref{eq:zero-alpha} follows  from  Proposition \ref{pro:decomposition}.
\end{proof}	
\begin{proof}[Proof of Corollary \ref{cor:ECE-indip}]
	Suppose that for each potential outcome $Y_D(w)$ with $D \subseteq V$ and $w=0,1$,
	there exist a partition $A$ and $B$ of $D \subseteq V$ such that $Y_A(w) $ and $Y_B(w)$,
	are conditionally independent given  $X_U=x_U$,  $x_U \in \mathcal{I}_U$.
	Then, we have that $\gamma_{D\mid x_U}(0)=\gamma_{D\mid x_U}(1)=0$ for Theorem 1 in \citet{rov-lup-lar:2013}. Then the result follows because, under the log-mean linear regression model in Equation \eqref{eq:LML-regression}, condition in Equation \eqref{eq:cor3} is verified.	
\end{proof}	

\section*{Acknowledgments}
We are grateful to  Ian Michael Paul  and  Tonya Sharp King  (Penn State University College of Medicine, USA)  
for providing us the data on the honey study, and to Fabio Picciafuochi (Azienda USL, Reggio Emilia, Italy) for providing us the data on  the morphine study. We also thank Luca La Rocca, Fabrizia Mealli and Alberto Roverato for helpful discussions. Alessandra Mattei acknowledges financial support from the Italian Ministry of Research and Higher Education through grant Futuro in Ricerca 2012  RBFR12SHVV$\!$\underline{$\;\,$}003.\vspace*{-8pt}

\bibliographystyle{chicago}
\bibliography{gamma-ref}
\end{document}